%% file: main-arxiv.tex
\documentclass[%
  cleveref,thm-restate,
  a4paper,UKenglish,
]{lipics-v2021}
\hideLIPIcs  %
\nolinenumbers %

\DeclareFontShape{T1}{lmr}{m}{scit}{<->ssub * lmr/m/scsl}{} %
\usepackage{graphicx}
\input{settings/graphics}
\graphicspath{./graphics/}%

\usepackage{xspace}

\usepackage{tabularray}
\UseTblrLibrary{booktabs}

\usepackage{cleveref}

\usepackage{mathtools}
\usepackage{amsthm}
\usepackage{thm-restate}
\DeclarePairedDelimiter{\set}{\{}{\}}
\DeclarePairedDelimiter{\Path}{\langle}{\rangle}
\DeclarePairedDelimiter{\abs}{\lvert}{\rvert}

\newcommand{\LL}{\mathcal{L}}

\newcommand{\OR}{\texttt{or}\xspace}
\newcommand{\AND}{\texttt{and}\xspace}
\newcommand{\NOT}{\texttt{not}\xspace}
\newcommand{\True}{\texttt{true}\xspace}
\newcommand{\False}{\texttt{false}\xspace}

\newcommand{\init}{\eta}

\newcommand{\Natural}{\mathbb{N}}

\newcommand{\pfdir}[1]{\proofsubparagraph{#1}}

\newcounter{problem}

\title{From FPT to W[P]: Classifying Zero Forcing, Power Domination and Their Variants}

\author{Max Göttlicher}%
{Karlsruhe Institute of Technology, Germany}%
{max.goettlicher@kit.edu}%
{https://orcid.org/0000-0002-5556-4140}%
{Supported by the German Research Foundation (DFG) as part of the Research Training Group GRK 2153 ``Energy Status Data - Informatics Methods for its Collection, Analysis and Exploitation''.}
\author{Lennard Hofmann}{Karlsruhe Institute of Technology, Germany}{}{}{}
\author{Christoph Niederbudde}{Karlsruhe Institute of Technology, Germany}{}{}{}

\authorrunning{M. Göttlicher, L. Hofmann, C. Niederbudde} %

\Copyright{Max Göttlicher, Lennard Hofmann, Christoph Niederbudde} %



\category{} %

\relatedversion{} %
\EventEditors{John Q. Open and Joan R. Access}
\EventNoEds{2}
\EventLongTitle{42nd Conference on Very Important Topics (CVIT 2016)}
\EventShortTitle{CVIT 2016}
\EventAcronym{CVIT}
\EventYear{2016}
\EventDate{December 24--27, 2016}
\EventLocation{Little Whinging, United Kingdom}
\EventLogo{}
\SeriesVolume{42}
\ArticleNo{23}

\begin{document}

  \maketitle

  \begin{abstract}
    \textsc{Zero Forcing (ZF)} and \textsc{Power Dominating Set (PDS)} mark vertices in a graph based on a common forcing process
    starting from a problem-specific set of initially marked vertices.
    ZF initially marks the selected vertices while PDS additionally marks their neighbors.
    In the forcing process, a marked vertex with only one unmarked neighbor may force that neighbor which then becomes marked, too.
    A solution marks the entire graph by exhaustive application of this rule.
    One variant generalizes the forcing threshold; vertices may force when they have a fixed number of $k$ unmarked neighbors.
    Another variant limits propagation to a fixed number of rounds.
    We classify the parameterized complexity of the problem variants obtained by combining these choices of initialization, round limit and forcing threshold.
    We show that with appropriate choices, these variants range in parameterized complexity from fixed-parameter tractable to complete for every even layer $W[2\ell]$ of the $W$-hierarchy, and up to $W[P]$-complete.
    Our results demonstrate that small changes in any one of these three dimensions can lead to a sharp change in problem complexity.
  \end{abstract}

  \section{Introduction}
  \label{sec:intro}
  \textsc{Dominating Set}, \textsc{Zero Forcing} and \textsc{Power Dominating Set} are three problems which
  ask for the smallest vertex set from which all vertices are marked by problem-specific marking rules.
  In case of \textsc{Dominating Set}, the marked vertices are the selected vertices and their neighbors.
  \textsc{Zero Forcing} marks vertices based on an iterative forcing process: initially, only the selected vertices are marked.
  Then, when a marked vertex has only one unmarked neighbor, the vertex forces its neighbor which becomes marked, too.
  This is repeated exhaustively.
  \textsc{Power Dominating Set} combines the two problems.
  A selection is a power dominating set if the selected vertices together with their neighbors are a zero forcing set.
  The difference to \textsc{Zero Forcing} is thus the initialization.

  While the forcing process is the same in \textsc{Power Dominating Set} and \textsc{Zero Forcing},
  the two problems emerged independently and with different motivations.
  \textsc{Zero Forcing} was first studied as a lower bound on the maximum nullity of
  a graph~\cite{aim_minimum_rank__special_graphs_work_group_zero_2008} and independently motivated in
  quantum physics~\cite{burgarth_full_2007,severini_nondiscriminatory_2008}.
  Computing a minimum zero forcing set is NP-complete~\cite{yang_fastmixed_2013}.
  When parameterized by the solution size \textsc{Zero Forcing} is fixed-parameter tractable~\cite{aazami_hardness_2008,bhyravarapu_parameterized_2025,scheffler_parameterized_2025}.

  \textsc{Power Dominating Set (PDS)}, on the other hand, was introduced to model the cost-efficient placement of measurement
  devices in electric networks~\cite{brueni_minimal_1993,haynes_domination_2002,brueni_pmu_2005}.
  It is also known as the \textsc{PMU Placement Problem} and its original definition is closely linked to
  the laws of electricity~\cite{haynes_domination_2002}.
  Like \textsc{Zero Forcing}, \textsc{PDS} is NP-complete~\cite{brueni_minimal_1993,haynes_domination_2002}. %
  In sharp contrast to \textsc{Zero Forcing}, \textsc{PDS} is $W[P]$-complete and thus unlikely to be fixed-parameter
  tractable in its solution size~\cite{blasius_efficient_2025}.

  \textsc{Power Dominating Set} and \textsc{Zero Forcing} have some natural variants obtained by modifying the initialization or the forcing process.
  The initialization determines the set of initially marked vertices.
  \textsc{Zero Forcing} simply maps $S \mapsto S$ and \textsc{PDS} maps $S \mapsto N[S]$.
  One can also allow a per-instance set of pre-marked vertices $P$ to obtain \textsc{Pre-Marked Zero Forcing} which maps $S \mapsto S \cup P$~\cite{cazals_power_2019}.
  An equivalent view on \textsc{Pre-Marked Zero Forcing} is that the vertices $P$ can be selected for free and do not count towards the solution size~\cite{aazami_hardness_2008}.

  We consider two independent modifications of the forcing process; limiting propagation rounds and increasing the propagation threshold.
  In monitoring applications one may want to limit the number of inference steps since they may accumulate measurement errors.
  This can be modeled in the forcing process by organizing all the forces into rounds.
  In each round, all possible forces are applied simultaneously.
  The problem variants \textsc{$\ell$-round Zero Forcing} and \textsc{$\ell$-round Power Dominating Set} impose a fixed limit on the number of propagation rounds~\cite{aazami_hardness_2008}.
  Increasing the forcing threshold allows forcing more than one neighbor.
  For a fixed number $k$, in \textsc{$k$-Power Dominating Set} and \textsc{$k$-Forcing}, a marked vertex can force if it has at most $k$ unmarked neighbors~\cite{amos_upper_2015}.

  The parameterized complexity of some variants has also been studied.
  Like the standard variant, \textsc{$\ell$-round Zero Forcing} is FPT~\cite{aazami_hardness_2008}.
  In contrast, with pre-marked vertices \textsc{Zero Forcing} becomes $W[2]$-hard~\cite{cazals_power_2019}.
  Some other hardness results are not explicitly stated in the literature but transfer from the standard problem in a straightforward way;
  this yields $W[2]$-hardness for \textsc{$\ell$-round PDS}~(cf.~\cite{guo_improved_2005,kneis_parameterized_2006}) and $W[P]$-completeness of \textsc{$k$-PDS} (cf.~\cite{blasius_efficient_2025}).

  Another related problem is \textsc{Target Set Selection (TSS)}~\cite{chen_approximability_2009}.
  TSS uses a different propagation process: an unmarked vertex becomes marked once it has sufficiently many marked neighbors.
  This is different from the forcing process in \textsc{Zero Forcing} and \textsc{Power Dominating Set} where any given marked vertex can only propagate to a limited number of neighbors.
  \textsc{TSS} is $W[P]$-complete~\cite{abrahamson_fixed-parameter_1995,bazgan_parameterized_2014} by reduction from \textsc{Monotone Circuit Satisfiability}, similar to the proof for the $W[P]$-completeness of \textsc{PDS}~\cite{blasius_efficient_2025}.
  When parameterized by treewidth, \textsc{TSS} is $W[1]$-hard~\cite{ben-zwi_treewidth_2011}, unlike \textsc{Zero Forcing} and \textsc{Power Dominating Set} which are FPT in that parameter~\cite{guo_improved_2005}.

  \begin{table}
    \caption{
      Complexity results of several configurations of \textsc{Generalized Forcing} when parameterized by solution size.
      Each cell shows the parameterized complexity for its combination of initialization, round limit and forcing threshold.
      Unless stated otherwise, the placement is exact.
      Our new results are marked in bold.
    }
    \label{tab:generalized-forcing-complexity}
    \NewTblrTableCommand\oldresult{}
    \newcommand{\newresult}[1]{{\bfseries\boldmath #1}}
    \NewTblrTableCommand\openproblem{\SetCell{font={\itshape}}}
    \crefname{corollary}{Cor.}{Cors.}
    \crefname{lemma}{Lem.}{Lems.}
    \crefname{theorem}{Thm.}{Thms.}

    \NewTblrTheme{nocaption}{
      \DefTblrTemplate{firsthead}{default}{}
      \DefTblrTemplate{middlehead}{default}{}
      \DefTblrTemplate{lasthead}{default}{}
    }
    \begin{talltblr}[
      theme=nocaption,
      note{a}={\textsc{$\ell$-round Pre-Marked $k$-Forcing} is in $W[2\ell]$ and $W[2(\ell-2)]$-hard for fixed $\ell \geq 3$.},
      note{b}={not explicitly shown; follows by straightforward extension of the construction for PDS},
    ]{
      colspec={ccXXX},row{1-1}={t,font=\sffamily},columns={c},
      hline{1,Z}={\heavyrulewidth},hline{3}={\lightrulewidth},hline{2}={3-5}{\cmidrulewidth,endpos,lr},
      cell{1}{1}={r=2}{h},
      cell{1}{2}={r=2}{h},
      cell{1}{3}={c=3}{c},
    }
      {Round\\Limit} & {Forcing\\Threshold} & Initialization \\
      & & {\textsc{ZF} \\ $S \mapsto S$} & {\textsc{Pre-Marked ZF} \\ $S \mapsto S \cup P$} & {\textsc{PDS} \\ $S \mapsto N[S]$}\\
      unbounded & $1$ & {FPT \\ see~\cite{aazami_hardness_2008}} & {\newresult {$W[P]$}\\ \cref{res:pre-marked-zf-hardness}} & {\oldresult {$W[P]$} \\ see~\cite{blasius_efficient_2025}} \\
      unbounded & {fixed $k > 1$} & {\newresult{$W[P]$}\\ \cref{res:k-forcing-wp}} & {\newresult {$W[P]$}\\ \cref{res:pre-marked-zf-hardness}} & \oldresult {$W[P]$\\ \cite{blasius_efficient_2025}\TblrNote{b}, \cref{res:kpds-wp-hardness}} \\
      {fixed $\ell$} & $1$ & {FPT+\newresult{p-kernel}\\ \cite{aazami_hardness_2008}, \cref{res:l-round-k-forcing}} & {\newresult{$W[2(\ell-2)]$-hard}\TblrNote{a} \\ \cref{res:pre-marked-zf-hardness}} & {\newresult {$W[2\ell]$}\\ \cref{res:w2l-complete}} \\
      {fixed $\ell$} & {fixed $k > 1$} & {\newresult{FPT+p-kernel}\\ \cref{res:l-round-k-forcing}} & {\newresult{$W[2(\ell-2)]$-hard}\TblrNote{a} \\ \cref{res:pre-marked-zf-hardness}} & {\newresult {$W[2\ell]$}\\ \cref{res:w2l-complete}} \\
    \end{talltblr}
  \end{table}

  The known results do not show how the three choices of initialization, round limit and forcing threshold interact in general.
  To study this problem space systematically, we introduce \textsc{Generalized Forcing} as a framework that captures all three choices.
  Our results classify the parameterized complexity of the resulting combinations, as summarized in \cref{tab:generalized-forcing-complexity}.
  Our hardness bounds are tight for all problems except the case of bounded propagation in pre-marked forcing where a small gap remains.
  We show that bounding propagation captures a hierarchy of problems between \textsc{Dominating Set} and \textsc{Power Dominating Set}:
  for every fixed $\ell \in \Natural$, \textsc{$\ell$-round Power Dominating Set} is $W[2\ell]$-complete.
  This renders \textsc{$\ell$-round Power Dominating Set} a family of natural complete problems for the even layers of the W-hierarchy.
  Only few natural problems with similar properties are known.
  For comparison, \textsc{Binary CSP} is complete for the odd layers $W[2d+1]$ when parameterized by treedepth $d$~\cite{bodlaender_parameterized_2023}.
  Introducing a round limit to \textsc{Zero Forcing} does not affect its fixed-parameter tractability;
  however, we show that this variant permits a more direct FPT algorithm and admits a polynomial kernel.
  To complement the kernelization result, we also note that another variant, \textsc{Connected Zero Forcing}~\cite{brimkov_complexity_2017,davila_bounds_2018}, which requires the selection to be connected, admits no polynomial kernel under standard assumptions.
  The forcing threshold also has a significant effect on the parameterized complexity; in particular, we show that \textsc{$k$-Forcing} is $W[P]$-complete for $k > 1$.

  The remainder of this paper is organized as follows.
  In \cref{sec:prelims}, we define the fundamental concepts and notions used throughout the paper.
  In \cref{sec:gf:def}, we introduce \textsc{Generalized Forcing} as a reduction tool that captures the different combinations of initialization, round limit and forcing threshold.
  We then derive tight lower and upper complexity bounds for \textsc{$\ell$-round PDS} and its variants in \cref{sec:w2l}.
  In \cref{sec:k-forcing}, we prove the $W[P]$-hardness of \textsc{$k$-Forcing} and bound the parameterized complexity of \textsc{Pre-Marked $\ell$-round Zero Forcing}.
  In \cref{sec:k-forcing-fpt}, we show that \textsc{$\ell$-round Zero Forcing} remains FPT with forcing threshold $k>1$ and show that this variant admits a polynomial kernel.
  We conclude the paper in \cref{sec:conclusion}.

  \section{Preliminaries}
  \label{sec:prelims}
  A simple, undirected \emph{graph} $G$ consists of a set of vertices $V(G)$ and a set of edges $E(G) \subseteq \binom{V}{2}$.
  As a shorthand notation, we also denote an edge $\set{u,v} \in E(G)$ by $uv$.
  Two vertices $u$ and $v$ are \emph{adjacent} if $uv \in E(G)$; both $u$ and $v$ are \emph{incident} to $uv$.
  The (open) neighborhood $N(v)$ of a vertex is the set of all vertices adjacent to $v$.
  The closed neighborhood $N[v] = N(v) \cup \set{v}$ additionally includes $v$.
  The \emph{degree} $d(v) = \abs{N(v)}$ of a vertex is the size of the open neighborhood.
  A \emph{leaf} is a vertex of degree one and an \emph{isolated vertex} has degree zero.
  We extend the definition of $N(v)$ and $N[v]$ to vertex sets by taking the union of all neighborhoods
  of vertices in the set, i.e. $N(X) = \bigcup_{v \in X} N(v)$ and $N[X] = \bigcup_{v \in X} N[v]$.
  Similarly, we extend the definition of functions $f: V \mapsto B$ to vertex sets.
  For $X \subseteq V$ we write $f(X) = \bigcup_{v \in X} f(v)$. %
  A clique $C \subseteq V(G)$ is a vertex set such that every pair of distinct vertices in $C$ is adjacent. %
  Removing a vertex $v$ yields $G - v$ by removing $v$ and all its incident edges.
  The \emph{induced subgraph} $G[V']$ for a vertex subset $V' \subseteq V(G)$ is obtained by removing all other vertices,
  i.e. $G[V'] = G - (V \setminus V')$.

  Two vertices $u, w \in V(G)$ are \emph{connected} if there is a \emph{path} from $u$ to $w$, i.e. a sequence
  $\langle u=v_0, \dots, v_\ell=w \rangle$ such that every vertex $v_i$ is adjacent to its predecessor $v_{i-1}$.
  A path is \emph{simple} if all vertices in the path are distinct.
  If only the start and end are the same, the path is a \emph{cycle}.
  A graph is \emph{acyclic} if it has no cycle.
  Connectivity is an equivalence relation and its equivalence classes are called \emph{connected components}.
  If removing a vertex $v$ increases the number of connected components, $v$ is a \emph{cutvertex}.

  We call two graphs $G$ and $H$ \emph{isomorphic} if there is a bijection $f: V(G) \mapsto V(H)$ such that
  $uv \in E(G)$ if and only if $f(u)f(v) \in E(H)$.
  Such a mapping is a \emph{graph isomorphism}.

  \subparagraph{Forcing Problems}
  \label{sec:prelims:forcing}
  For a graph $G$ and an integer $k$, the forcing process expands an initial set $M_1 \subseteq V(G)$ of marked vertices over the
  course of several rounds.
  In each round $i$, every vertex marked in $M_i$ with at most $k$ unmarked neighbors \emph{forces} all of its
  unmarked neighbors, i.e. they too become marked in $M_{i+1}$.
  We call $k$ the \emph{forcing threshold}.
  More formally,
  \begin{equation*}
      M_{i+1} = M_i \cup \set[\big]{v \in N(w) \mid w \in M_i \text{ and } \abs{N(w) \setminus M_i} \leq k}\text{.}
  \end{equation*}
  We call an initial set $M_1$ an \emph{$\ell$-round $k$-forcing set} of $G$ if all vertices in $G$ are marked in $M_\ell$,
  i.e. $M_\ell = V(G)$.
  Since at least one new vertex is marked in every successful propagation round, $M_i$ reaches a fix-point in $M_{\abs{V}} = M_\infty$.
  We call an initial set $M_1$ with $M_\infty=V(G)$ a \emph{$k$-forcing set} of $G$.
  By this convention, we call round one the \emph{initialization} and all successive rounds \emph{propagation} rounds.
  The number $\ell$ is the \emph{round limit}. %
  As a convention, we omit the round limit in the problem name if $\ell>\abs{V}$; this is equivalent to an infinite number of rounds.

  Beside the forcing rule, the forcing problems considered in this paper differ by their \emph{initialization}.
  The initialization is a map $\init: V(G) \to 2^{V(G)}$ describing which vertices become initially marked by selecting a vertex $v$.
  For a given selection $S \subseteq V(G)$, we set the initially marked set $M_1 = \init(S)$, where $\init(S) = \bigcup_{v \in S} \init(v)$.
  We call the resulting framework \emph{\textsc{Generalized Forcing}}.
  The following problems are special cases.
  With \emph{\textsc{PDS}-type} initialization $\init(v) = N[v]$, we get \textsc{Power Dominating Set};
  with \emph{ZF-type} initialization $\init(v) = v$, we get \textsc{Zero Forcing} as the base problems with unbounded propagation and forcing threshold $1$.
  Introducing a round limit $\ell$, generalized forcing threshold $k$, or both gives \textsc{$\ell$-round Power Dominating Set},
  \textsc{$k$-Power Dominating Set}, \textsc{$\ell$-round $k$-Power Dominating Set}, \textsc{$\ell$-round Zero Forcing}, \textsc{$k$-Forcing}, and \textsc{$\ell$-round $k$-Forcing}.

  \subparagraph{Parameterized Complexity}
  \label{sec:prelims:complexity}
  We give only a brief introduction to parameterized complexity; for more details, we refer the interested reader to a
  textbook on the subject~\cite{downey_fundamentals_2013}.
  A \emph{parameterized problem} is a language $\LL \subseteq \Sigma^\ast \times \Natural$ with instances $(x,k)$ where $k$ is called the \emph{parameter}.
  A \emph{parameterized reduction} from $\LL$ to $\LL'$ is an algorithm that maps $(x,k)$ to $(x', k')$ such that
  \begin{itemize}
    \item $(x, k) \in \LL$ if and only if $(x', k') \in \LL'$
    \item $k' \leq f(k)$
    \item the running time of the transformation is bounded by $g(k) \cdot \abs{x}^{O(1)}$
    \item $f,g$ are computable functions.
  \end{itemize}
  A \emph{kernel} is an algorithm that maps an instance $(x,k)$ to an equivalent instance $(x', k')$ of the same problem
  such that $\abs{x'} + k' \leq f(k)$ for some computable function $f$ and runs in polynomial time.
  A kernel is \emph{polynomial} if $\abs{x'} + k' \leq k^{O(1)}$.

  A parameterized problem $\LL$ is \emph{fixed parameter tractable} if there is an
  algorithm deciding $\LL$ in time $f(k) \cdot \abs{x}^{O(1)}$ where $f$ is a computable function.
  If $\LL$ is fixed parameter tractable and there is a parameterized reduction from $\LL'$ to $\LL$ then $\LL'$ is also
  fixed parameter tractable.

  The $W$-hierarchy consists of classes $\text{FPT} \subseteq W[1] \subseteq W[2] \subseteq \dots \subseteq W[P]$.
  A central conjecture of parameterized complexity is that these inclusions are strict.
  Therefore, if a problem is hard for some $W[t]$, that problem is assumed not to be FPT.  %
  Many complete problems are known for $W[1]$ and $W[2]$, e.g. the parameterized versions of \textsc{Independent Set}
  and \textsc{Dominating Set} which are $W[1]$-complete and $W[2]$-complete when parameterized by their solution size, respectively.

  The $W$-hierarchy is characterized in terms of \emph{boolean circuits},
  directed acyclic graphs where each vertex is either an \OR, \AND, or \NOT gate.
  The sources of a circuit are \emph{inputs} and the single sink is the \emph{output}.
  In \emph{monotone} circuits, all inputs are variables $x_i$ and
  in \emph{antimonotone} circuits all inputs are negated variables $\neg x_i$.
  A \emph{$t$-normalized} circuit consists of an input layer followed by $t$ alternating layers of \OR- and \AND- gates between the inputs and the output which is always an \AND-gate.
  For even $t$, \textsc{Weighted Monotone $t$-Normalized Circuit Satisfiability (WMNS[$t$])} is $W[t]$-complete and
  for odd $t$, \textsc{Weighted Antimonotone $t$-Normalized Circuit Satisfiability (WANS[$t$])} is $W[t]$-complete~\cite{downey_fixed-parameter_1995}.

  \section{The Generalized Forcing Framework}
  \label{sec:gf:def}

  This section introduces generalized forcing which we use as a reduction device throughout the paper.
  The purpose of this section is two-fold.
  First, we define \textsc{$\ell$-round Generalized $k$-Forcing} as a generalization of \textsc{Power Dominating Set} and \textsc{Zero Forcing} with arbitrary initialization maps.
  We show that three useful extensions, implication arcs, optional vertices and pre-marked vertices,
  can be eliminated through parameterized transformations preserving the solution size.
  In later parts of this paper, we can therefore freely use these extensions.
  Second, we show how arbitrary initialization maps can be expressed in terms of power domination.
  Consequently, hardness results obtained for generalized forcing translate to \textsc{Power Dominating Set} and its variants.

  A basic instance of \emph{\textsc{$\ell$-round Generalized  $k$-Forcing}} consists of a graph $G$, an initialization $\init$ and solution size $d$.
  We define \emph{\textsc{Generalized $k$-Forcing}} as a shorthand for the variant with at least $\abs{V}$ rounds.
  The problem asks whether there is a vertex selection $S \subseteq V(G)$ of size $\abs{S} \leq d$
  such that $\init(S)$ is an $\ell$-round $k$-forcing set of $G$ where $\ell$ may be infinite.
  This basic variant is a minimal extension of \textsc{PDS} and \textsc{Zero Forcing}.

  For every fixed $k$, \textsc{Generalized $k$-Forcing} is $W[P]$-hard by reduction from the $W[P]$-complete \textsc{Power Dominating Set}.
  It is also contained in $W[P]$ since a nondeterministic Turing machine can guess a solution of size $d$ and then verify it deterministically in polynomial time~\cite{cai_parameterized_1997}.
  \begin{restatable}{lemma}{REgfwphardness}
    \label{res:gf:wp-hardness}
    Every instance of \textsc{Power Dominating Set} can be transformed to an equivalent instance of \textsc{Generalized $k$-Forcing} in polynomial time.
    The transformation preserves the solution size.
  \end{restatable}
  \begin{proof} %
    The idea here is that we saturate the forcing threshold by attaching $k-1$ new leaves to each vertex $v \in V$.
    The leaves have empty initialization; selecting them has no effect.
    All original vertices keep PDS-type initialization.
    Before $v$ forces, these leaves thus remain unmarked; therefore, $v$ can only force one additional original neighbor.
  \end{proof}

  \subsection{Generalized and Extended Instances}
  \label{sec:gf}

  For the reductions it will be convenient to work with \emph{extended instances}.
  Such instances have \emph{optional vertices} $O \subseteq V(G)$, directed
  \emph{implication arcs} $A \subseteq V(G) \times V(G)$ and \emph{pre-marked vertices} $P$.
  Optional vertices do not need to be marked by a solution while pre-marked vertices are marked regardless of the vertex selection.
  If there is an implication arc from $u$ to $w$ and $u$ becomes marked in round $j$, then $w$ becomes marked in round $j+2$.

  We show that every extended instance can be transformed by a parameterized reduction to an equivalent basic instance
  without optional vertices, implication arcs or pre-marked vertices.
  We use the gadgets in \cref{fig:gf-extension-gadgets}, to eliminate the extensions in three steps.
  First, we replace implication arcs, then optional vertices, and finally pre-marked vertices.

  \begin{figure}
    \centering
    \subcaptionbox{
      The timer gadget is used to mark the optional vertex $v \in O$ in round $\ell$ when propagation is bounded.
      \label{fig:no-targets:timer}
    }[0.49\linewidth]{
      \begin{tikzpicture}
        \graph[] {
            {{[nodes={node,marked,private leaf angle=90}]t1/,/,td/}} --
            {{[nodes=node]t2/,/,/}} --
            {{/$\dots$,/$\dots$,/$\dots$}} --
            {{[nodes=node]tl/,/,/}} --
            {{/[hidden],tlm/["$t_v$" left, node,marked,private leaf angle=0,],/[hidden]}} --
            {{/[hidden],v/["$v$"left,node],/[hidden]}}
        };
        \draw[decorate, decoration={brace, raise=4mm}] (tl.south) -- node[left=5mm] {$\ell-1$} (t1.north);
        \draw[decorate, decoration={brace, raise=4mm}] (t1.west) -- node[above=5mm] {$k+1$} (td.east);

      \end{tikzpicture}
    }
    \hfill
    \subcaptionbox{\label{fig:no-targets:and}
    Gadget for marking the optional vertex $v$ after all non-optional vertices $V \setminus O$ are marked with unbounded propagation.
    }[0.49\linewidth]{
      \begin{tikzpicture}
        \graph[] {
            {[nodes={node}] {[nodes={not source}] v1/,/[dots],vt/}, {[nodes=objective] vt1/,/[dots],vi/["$v$" left],/[dots],vn/}, ll/[hidden,c]}
          --[->-=0.7,matching] {[nodes={node}]/,/[dots],/,}
          -- {/[hidden], bottom/["$a_v$" left,node,marked,private leaf angle=-90,private leaves]}
          --[hv] vi;
        };
        \node[draw, rectangle, rounded corners, fit=(v1)(vn), inner sep=2mm] (g) {};
        \node also["$G$" left] (g);
        \draw[decorate, decoration={brace, raise=4mm}] (v1.west) -- node[above=5mm] {$V \setminus O$} (vt.east);
        \draw[decorate, decoration={brace, raise=4mm}] (vt1.west) -- node[above=5mm] {$O$} (vn.east);
      \end{tikzpicture}
    }
    \subcaptionbox{\label{fig:no-targets:implication}
    Gadget for implication arcs. %
    Pre-marked vertices with $k-1$ leaves are depicted as \legendnode{node,marked,private leaves}.
    }[\linewidth]{
      \begin{tikzpicture}
        \graph[grow right=3, branch down,] {
          s1/["$u$",node] --[->-] t1/["$v$",node];
        };
        \graph[grow right=0.75, branch down=0.5,nodes={shift={(6,0.5)}}] {
            {/[c,hidden], s2/[node,"$u$"]}
          -- {a/[node,marked,private leaves], /[vdots,hidden], b/[node,marked,private leaf angle=-60,private leaves]}
          -- {/[c], /{$k+1$}[hidden], /[c]}
          -- {a1/[node], /[hidden,vdots], b2/[node]}
          -- {/[c,hidden], out/[node,marked]}
          -- {/[c,hidden], t2/[node,"$v$"]};
        };
        \draw[->,decorate,decoration={snake,post length=1mm,amplitude=1pt,segment length=2mm}] ($(t1) + (8mm,0)$) -- ($(s2) - (8mm,0)$);
      \end{tikzpicture}
    }
    \caption{
      Gadgets to eliminate extensions from instances of \textsc{Generalized Forcing}.
      We represent pre-marked vertices by \legendnode{node, marked} and use \legendnode{node,draw=none,fill=none,private leaves}
      to indicate that a node has $k-1$ leaves attached.
    }
    \label{fig:gf-extension-gadgets}
  \end{figure}
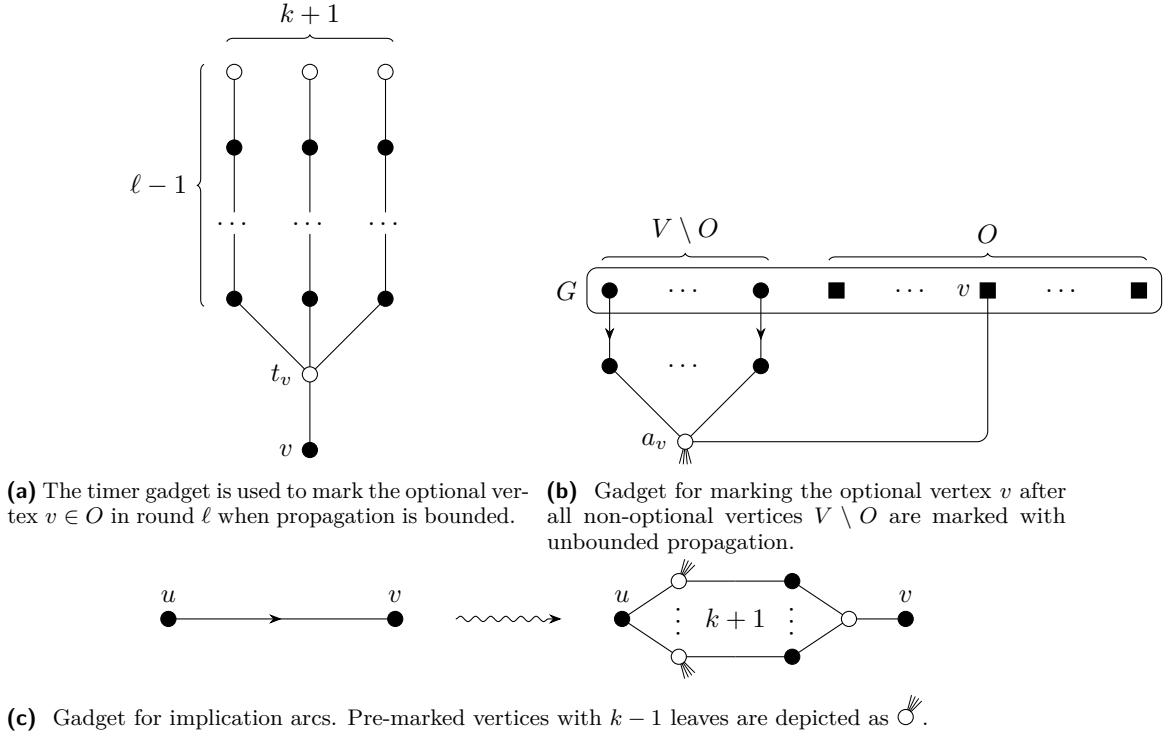

  \subparagraph*{Implication Arcs}
  \label{sec:gf:extended:implication-arcs}
  We replace the implication arc by a gadget that allows propagation only in one direction.

  \begin{restatable}{lemma}{REimplicationarcs}
    \label{res:implication-arcs}
    Every extended instance of \textsc{$\ell$-round Generalized $k$-Forcing} can be transformed into an equivalent instance without
    implication arcs in polynomial time while preserving the solution size.
  \end{restatable}
  \begin{proof}
    Let $I$ be an extended instance of \textsc{$\ell$-round Generalized $k$-Forcing} with graph $G$, implication arcs $A$.
    We construct an equivalent instance $I'$ from $I$ by replacing each implication arc $u \to v \in A$ by the gadget shown in \cref{fig:no-targets:implication}.
    The gadget replaces the arc by $k+1$ paths $\Path{u, a_i, b_i, c}$
    from $u$ to a new pre-marked vertex $c$ which in turn has an edge to $v$.
    The vertices $a_i$ are pre-marked and have $k-1$ leaves each.
    All gadget vertices are optional.
    The initialization of all gadget vertices is left empty; all other vertices use the initialization given by $I$.
    Since $u$ and $v$ only gain pre-marked neighbors, their propagation threshold remains unaffected.

    \pfdir{$u \Rightarrow v$}
    If $u$ becomes marked in round $j$, each $a_i$ has only $k$ unmarked neighbors left which it can force in round $j+1$.
    Then, $c$ has at most one unmarked neighbor left which it can thus force in round $j+2$: $v$.

    \pfdir{$v \nRightarrow u$}
    Conversely, suppose $u$ is not marked.
    The leaves on each $a_i$ and all $b_i$ cannot be marked by the initialization.
    In order to be marked, they need to be forced by either some $a_i$ or by $c$.
    Since $u$ is unmarked, each $a_i$ has $k+1$ unmarked neighbors and $c$ is adjacent to all $k+1$ vertices $b_i$; neither $a_i$ nor $c$ can force.
    Therefore, $u$ cannot be marked by forcing from the gadget.

    \pfdir{Instance Equivalence}
    Given a solution $S$ of the original instance, each implication arc is now simulated exactly by its gadget.
    Since all vertices in the original instance are marked by $S$, all non-optional vertices in the new instance are also marked.

    Conversely, since the gadget vertices have empty initialization, $S = S' \cap V(G)$ must already be a solution of $G'$.
    Since the gadgets can be replaced by implication arcs, $S$ must also be a solution of $G$.
  \end{proof}

  We need the internal vertices to be optional since, otherwise, we cannot easily ensure they are marked.
  If, however, the input graph contains no optional vertices and the round limit is unbounded,
  the internal vertices are also guaranteed to become marked by every solution.
  Therefore, in this case, the internal vertices do not need to be optional.
  We formalize this in the following corollary.
  \begin{restatable}{corollary}{REimplicationarcsunbounded}
    \label{res:implication-arcs-unbounded}
    Every extended instance of \textsc{Generalized $k$-Forcing} without optional vertices can be transformed to an equivalent
    extended instance without optional vertices or implication arcs in polynomial time while preserving the solution size.
  \end{restatable}
  \subparagraph*{Optional Vertices}
  \label{sec:gf:extended:optional}
  We can eliminate all optional vertices using implication arcs and pre-marked vertices.
  To account for large or unbounded round limits, we use two distinct constructions.
  In case $\ell$ is bounded, we use a timer gadget that marks all vertices after a certain number of rounds.

  \begin{restatable}{lemma}{REtargetverticesbounded}
    \label{res:target-vertices-bounded}
    Every extended instance of \textsc{$\ell$-round Generalized $k$-Forcing} with finite round limit and with optional vertices can
    be transformed into an equivalent instance without optional vertices in polynomial time while preserving the solution size.
  \end{restatable}
  \begin{proof}
    Let $I$ be an extended instance of \textsc{$\ell$-round Generalized $k$-Forcing} with graph $G$ and optional vertices $O$.
    If $O=V(G)$, the instance is trivially solved by the empty set; we output an empty graph in this case.
    If $\ell=1$, the only round is the initialization; an equivalent instance is $G - O$.
    Thus, assume in the following that $V(G) \setminus O$ is non-empty and that $\ell > 1$.

    We construct $I'$ from $I$ by appending a timer gadget to each optional vertex as shown in \cref{fig:no-targets:timer}
    to each optional vertex $v \in O$ as follows.
    \begin{itemize}
      \item Attach a new vertex $t_v$ to $v$
      \item for each $j \in \set{1, \dots, k+1}$ create a path $\Path{\rho(v, j, 1), \dots, \rho(v, j, \ell-1)}$
      \item add edges from each $\rho(v, j, \ell-1)$ to $t_v$
      \item pre-mark $\rho(v, 1, 1)$, $\rho(v, 2, 1)$, \dots, $\rho(v, k+1, 1)$, and $t_v$
    \end{itemize}
    All new gadget vertices have empty initialization.
    No vertex in $G'$ is optional.
    The resulting graph $G'$ can clearly be constructed in polynomial time.

    \pfdir{Solution in $I$ $\Rightarrow$ Solution in $I'$}
    Let $S$ be a solution of $G$.
    Then $S$ is also a solution of $G'$.

    The initialization is unchanged for $v \in V(G)$, thus all vertices marked by the initialization in $G$ are also marked in $G'$.
    Furthermore, no vertex $v \in V(G)$ has an edge to an unmarked vertex in a timer gadget.
    The timer gadgets become marked by propagation from $\rho(v,j,1)$ and thus each $\rho(v, j, i)$ becomes marked in round $i$.
    Thus, for each optional $v \in O$, $t_v$ has only $k$ unmarked neighbors left in round $\ell-1$, and can
    force $v$ to become marked in the final round $\ell$.
    The set of unmarked neighbors is the same for all $v \in V(G)$ in the first $\ell-1$ rounds
    and thus propagation outside the timer gadgets is identical.

    All optional vertices in $G'$ are marked in round $\ell$ and thus $S$ is a solution of $G'$.

    \pfdir{Solution in $I'$ $\Rightarrow$ solution in $I$}
    Let $S'$ be a solution of $G'$.
    Then $S = S' \cap V(G)$ is a solution of $G$.

    Each $t_v$ has $k+1$ neighbors in the timer gadget: $k+1$ path vertices $\rho(v, j, \ell-1)$.
    The initialization of the gadget vertices is restricted to their neighborhood.
    Since $t_v$ is also adjacent to $v$, any selection of gadget vertices leaves at least $k+1$ neighbors of $t_v$
    unmarked before round $\ell-1$.
    Thus, $t_v$ can never force before round $\ell-1$.

    The timer gadgets cannot force any vertex outside the gadget before the last round.
    Thus, by the same argument as for the other direction, before round $\ell$ a vertex $v \in V(G)$ outside the gadget
    is marked in $G$ if and only if $v$ is also marked in $G'$.

    Since the timer gadgets can only mark optional vertices in round $\ell$, all other vertices must have been marked by
    the initialization and normal propagation rounds.
    Thus $S \cap V(G)$ must be a solution of $G$.
  \end{proof}

  With unbounded rounds, the size of the timer gadget would be infinite, so we need a different approach.
  We exploit the fact that that propagation through the graph cannot take more than $\abs{V}$ rounds.
  Our construction then adds a gadget that marks the optional vertices after non-optional vertices have been marked,
  so all vertices are marked in at most $\abs{V} + 3$ rounds.

  \begin{restatable}{lemma}{REtargetverticesunbounded}
    \label{res:target-vertices-unbounded}
    Every extended instance of \textsc{Generalized $k$-Forcing} with optional vertices can be transformed into an equivalent
    instance without optional vertices in polynomial time while preserving the solution size.
  \end{restatable}
  \begin{proof}
    Let $I$ be an extended instance of \textsc{Generalized $k$-Forcing} with graph $G$ and optional vertices $O$.
    We construct an equivalent instance $I'$ without optional vertices.
    Like in \cref{res:target-vertices-bounded}, if $O=V(G)$, the instance is trivially solved by an empty selection;
    we output an empty graph.
    Thus, assume in the following that $V \setminus O$ is non-empty.

    We construct $I'$ from $I$ by attaching a cleanup gadget to each optional $v \in O$ as depicted in \cref{fig:no-targets:and}.
    For each optional $v \in O$ we create a new pre-marked \emph{collector} vertex $a_v$ adjacent to
    $v$ and with $k-1$ new leaves attached.
    Then, for each non-optional $w \in V(G) \setminus O$, we create a new \emph{guard} vertex $\sigma(v, w)$ adjacent to $a_v$.
    Add an implication arc from $w$ to $\sigma(v, w)$.
    All gadget vertices have empty initialization; the initialization of the original vertices remains unchanged.

    Let $S$ be a solution of $I$.
    Then $S$ eventually marks all non-optional vertices $v \in V(G) \setminus O$.
    Since $G'$ adds only implication arcs and pre-marked neighbors to vertices in $G$, all vertices that become marked in $G$ also become marked in $G'$.
    As soon as all non-optional vertices are marked, each $a_v$ has at most $k$ unmarked neighbors left which it can force, in particular $v$.
    Thus all gadget vertices and all optional vertices become marked.
    Hence, $S$ is a solution of $G'$.

    Conversely, let $S$ be a solution of $G'$.
    Since the gadget vertices have empty initialization, $S \cap V(G)$ is a solution of $G'$.
    Each collector $a_v$ can only force $v$ once all guard vertices $\sigma(v, w)$ are marked.
    This is only the case when all non-optional vertices are marked; denote the first round when this is the case by $\ell$.
    But then all forces in the first $\ell$ rounds that happen in $G'$ can also happen in $G$.
    Thus all non-optional vertices in $G$ become marked in the first $\ell$ rounds.
    Since $I$ does not require optional vertices to be marked, $S \cap V(G)$ is a solution of the original instance.
  \end{proof}

  \begin{restatable}{corollary}{REtargetvertices}
    \label{res:target-vertices}
    Every extended instance of \textsc{Generalized Forcing} can be transformed into an equivalent instance without optional
    vertices in polynomial time while preserving the solution size.
  \end{restatable}

  \subparagraph*{Pre-Marked Vertices}
  \label{sec:gf:extended:pre-marked}
  The last remaining extension is the set of pre-marked vertices.
  We incorporate them into the initialization.
  \begin{restatable}{lemma}{REpremarkedelimination}
    \label{res:pre-marked-elimination}
    Every extended instance of \textsc{$\ell$-round Generalized $k$-Forcing} can be transformed into an equivalent
    instance without pre-marked vertices in polynomial time while preserving the solution size.
  \end{restatable}
  \begin{proof}
    Instances with pre-marked vertices may admit an empty solution.
    Whether that is the case can easily be verified by simulating the forcing process starting only with the pre-marked vertices.
    If the instance admits an empty solution, we chose the empty graph as equivalent instance.

    Otherwise, we can trivially obtain an equivalent instance by adding the pre-marked vertices to the initialization of every vertex.
    This does not affect the underlying graph and no new optional vertices are introduced so forcing remains unaffected.
  \end{proof}

  Since \cref{res:implication-arcs} introduces only new pre-marked vertices and \cref{res:pre-marked-elimination}
  introduces no extended features, we can apply \cref{res:target-vertices,res:implication-arcs,res:pre-marked-elimination}
  to obtain a basic instance without extensions.
  \begin{restatable}{lemma}{REgeneralizedforcingext}
    \label{res:generalized-forcing-extension}
    Every extended instance of \textsc{$\ell$-round Generalized $k$-Forcing} can be transformed into an equivalent basic
    instance in polynomial time while preserving the solution size.
  \end{restatable}

  \subparagraph*{Extended ZF-type instances}
  The gadgets require empty initialization in the inner gadget vertices.
  This is incompatible with ZF-type instances, in which every vertex always marks itself.
  However, this property also gives us more freedom to exchange gadget vertices in a solution for original vertices.
  Indeed, for ZF-type instances, we can show that we can obtain an equivalent solution that does not contain gadget vertices.

  \begin{restatable}{lemma}{REextendedzfinstances}
    \label{res:extended-zf-instances}
    Every extended instance of \textsc{$\ell$-round $k$-Forcing} with ZF-type initialization can be
    transformed in polynomial time into an equivalent instance without implication arcs or optional vertices but with pre-marked vertices.
    The transformation preserves the solution size.
  \end{restatable}
  \begin{proof}
    The previous proofs assume that selecting gadget vertices has no effect.
    It suffices to show that we can replace selected gadget vertices with non-gadget vertices.
    The proofs then remain applicable.

    In the gadget for an implication arc from $u$ to $v$, if a single internal gadget vertex is selected, replace it with $u$.
    If two or more internal gadget vertices are selected, replace all of them with $u$ and $v$.

    In the timer gadget to select an optional vertex $v$, replace all selected internal gadget vertices with $v$.

    In the cleanup gadget for an optional vertex $v$ with unbounded propagation, replace a selected guard vertex $\sigma(v, w)$ with its non-optional vertex $w$.
    Selecting $a_v$ has no effect since it is pre-marked; delete it from the selection.
    If a leaf of $a_v$ is selected, replace it with $v$.

    In each case, the solution size does not increase.
  \end{proof}

  \subsection{Reducing Generalized Forcing to Power Domination}
  \label{sec:gf:pds}
  We introduced \textsc{Generalized $k$-Forcing} as a reduction framework in which the initialization can be chosen freely.
  We show now that we can restrict this freedom by modifying the underlying graph.
  We show that every generalized instance can be transformed into an equivalent instance with \textsc{PDS}-type initialization.
  Therefore, the hardness results transfer to \textsc{Power Dominating Set}.

  \begin{restatable}{lemma}{REpdshardness}
    \label{res:pds-hardness}
    Every instance of \textsc{$\ell$-round Generalized $k$-Forcing}, where $\ell$ may be infinite, can be transformed
    into an equivalent instance of \textsc{$\ell$-round $k$-Power Dominating Set} in polynomial time while preserving the solution size $d$.
  \end{restatable}
  \begin{proof}
    Let $I$ be an instance of \textsc{$\ell$-round Generalized $k$-Forcing} with graph $G$, initialization $\init$ and solution size $d$.
    We may assume that $d \leq \abs{V}$ and that $I$ has no empty solution.
    For $d > \abs{V}$ we can trivially decide the problem by selecting all vertices
    and verifying whether all vertices become marked.

    Otherwise, we construct an equivalent instance $I'$ with graph $G'$.
    The idea here is to simulate an arbitrary initialization map through the domination rule.
    To this end, we transform the instance using the construction shown in \cref{fig:generalized-forcing-to-pds}.
    We introduce a proxy vertex $\rho(v)$ for each original vertex $v \in V(G)$.
    To ensure that only proxy vertices can be selected, we create $d+k+1$ disjoint copies $G_j$ of $G$.
    Each proxy vertex $\rho(v)$ then has edges to the vertices corresponding to $\init(v)$ in each copy $G_j$.
    When selected, $\rho(v)$ then marks these vertices by domination.
    There are no edges between different copies $G_j$ and $G_{j'}$.
    We ensure that all proxies become marked if any of them is selected by making them a clique $C$.

    More formally, we construct $G'$ as follows.
    \begin{enumerate}
      \item Create $d+k+1$ disjoint copies $G_j$ of $G$.
      For each $v \in V$, denote by $\sigma(v, j)$ the copy of $v$ in $G_j$.
      \item Insert a clique $C = \set{\rho(v) \mid v \in V}$ of $\abs{V}$ vertices .
      \item For each $v \in V$, each $w \in \init(v)$ and each copy $G_j$, add an edge $\rho(v)\sigma(w, j)$.
    \end{enumerate}
    Since $k$ is fixed and $d \leq \abs{V(G)}$, $G'$ can be constructed in polynomial time.

    \pfdir{{Generalized Forcing} to {Power Domination}:} %
    Let $S$ be a solution of $I$.
    We show that $S' = \rho(S)$ is a solution of $I'$.

    Since $S$ is not empty, the selected proxy vertices mark the entire clique $C$.
    Furthermore, in each copy $G_j$ they dominate exactly the vertices $\sigma(\init(S), j)$ which are thus marked in the first round.
    After the initialization, all vertices outside the copies are marked and the unmarked neighbor counts inside the copies are the same as in $G$.
    Thus, propagation in the copies is identical to propagation in $G$.
    Since $S$ solves the original instance, $S'$ marks all vertices in each copy $G_j$ and is thus a solution of $I'$.

    \pfdir{Power Dominating Set to Generalized Forcing:}
    For the reverse direction, we first show that for a given solution $S'$ of $I'$, restricting $S'$ to $C$ still yields a solution.
    Mapping that restricted solution back to $G$ then yields a solution of $I$.

    Since $\abs{S'} \leq d$, at least $k+1$ copies of $G$ contain no selected vertex.
    Let $J$ be a set of indices of $k+1$ such copies.
    We prove by induction over the rounds that for two unselected copies $j_1, j_2 \in J$ and a vertex $v \in V(G)$,
    the vertex $\sigma(v, j_1)$ is marked if and only if $\sigma(v, j_2)$ is marked in a given round.
    Furthermore, no vertex inside $G_j$ with $j \in J$ is forced from a vertex in $C$.

    In the first round, no vertex in any copy $j \in J$ is selected.
    Thus, the only way for a vertex in such a copy to become marked is by domination from $C$.
    Since all copies are identically connected to $C$, corresponding vertices in all copies are marked in the first round.

    Now assume the claim holds up to some round $i$.
    Then, by our assumption, all corresponding vertices in the copies $J$ have corresponding marked and unmarked neighbors inside their copies.
    Hence, the same internal forces are possible in each copy $j \in J$.
    Furthermore, if a vertex $\rho(v) \in C$ has an unmarked neighbor $\sigma(w, j)$ with $j \in J$, then the corresponding vertices $\sigma(w, j')$ are unmarked in all unselected copies $j' \in J$.
    Thus, $\rho(v)$ has at least $k+1$ unmarked neighbors and cannot force.
    The induction claim follows for round $i+1$.

    Since $S'$ is a solution of $I'$, every vertex in each unselected copy $G_j$ with $j \in J$ becomes eventually marked.
    By the induction above, all forces after the first round occur only within $G_j$.
    Let $S = \rho^{-1}(S' \cap C)$.
    The vertices dominated in $G_j$ by $S' \cap C$ are $\sigma(\init(S), j)$.
    Since the copies are isomorphic to $G$, propagation within any copy $G_j$ is identical to propagation in $G$ starting from $\init(S)$.
    Thus, $S$ is a solution of $I$.
  \end{proof}
  Together with \cref{res:gf:wp-hardness} and the $W[P]$-hardness of \textsc{PDS}, this reduction yields an explicit hardness result for \textsc{$k$-PDS}.
  \begin{restatable}{corollary}{REkpdswp}
    \label{res:kpds-wp-hardness}
    \textsc{$k$-Power Dominating Set} is $W[P]$-complete.
  \end{restatable}

  \begin{figure}
    \centering
    \begin{tikzpicture}[flat node/.style={node}]%
      \graph {
        rv1/["$\rho(v_1)$" above,node], /[dots],
        rvi/["$\rho(v_i)$" above,node], /[dots],
        rvn/["$\rho(v_n)$" above,node]
      };
      \foreach \n/\j/\a in {1/1/-1,j/j/0,{d+k+1}/n/1}{
        \begin{scope}[shift={($(rvi) + (0,4) + ({270 + \a*30}:8cm)$)}]
          \node[flat node] (a) at (0.75,-0.2) {};
          \node[flat node] (b) at (-0.63,0.67) {};
          \node[flat node] (c) at (0.45,1.09) {};
          \node[flat node] (d) at (1.4,0.7) {};
          \node[flat node] (e) at (0.5,0.4) {};
          \node[flat node] (f) at (-0.1,0.1) {};
          \graph[use existing nodes]{
            e -- b -- c -- d -- e -- {a,f};
          };
        \end{scope}
        \node[fit=(a)(b)(c)(d)(e),draw,ellipse,inner sep=-1.5mm,"$G_{\n}$" below] (g\j) {};
        \draw[thin,very nearly transparent] foreach \t in {rv1,rvn}{
          let \p1=(g\j), \p2=(\t), \n1={atan2(\y2-\y1, \x2-\x1)} in {
            foreach \b in {-10,0,10}{
              let \n2={\n1+\b} in {
                (g\j.\n2) -- (\t)
              }}}
        };
      }
      \draw foreach \t in {rvi}{ foreach \j in {1,j,n}{
        let \p1=(g\j), \p2=(\t), \n1={atan2(\y2-\y1, \x2-\x1)} in {
          foreach \b in {-10,0,10}{
            let \n2={\n1+\b} in {
              (g\j.\n2) -- (\t)
            }}}}};
      \foreach \t in {rvi}{ \foreach \n/\j in {1/1,j/j,{d+k+1}/n}{
        \path (\t) -- node[fill=white,fill opacity=0.8,rounded corners,text opacity=1,near end,inner sep=1mm] {$\sigma(\init(v_i),\n)$} (g\j);
      }}
      \node[fit=(rv1)(rvi)(rvn), rectangle, rounded corners, draw, "$C$" left] {};
      \path[] (g1) -- node[sloped] {$\dots$} (gj) --node[sloped] {$\dots$} (gn);
    \end{tikzpicture}

    \caption{
      Reduction from \textsc{($\ell$-round) Generalized $k$-Forcing} to \textsc{($\ell$-round) $k$-Power Dominating Set}. We denote by $\sigma(v,j)$ the vertex corresponding to $v$ in the $j$-th copy of the graph.
    }
    \label{fig:generalized-forcing-to-pds}
  \end{figure}
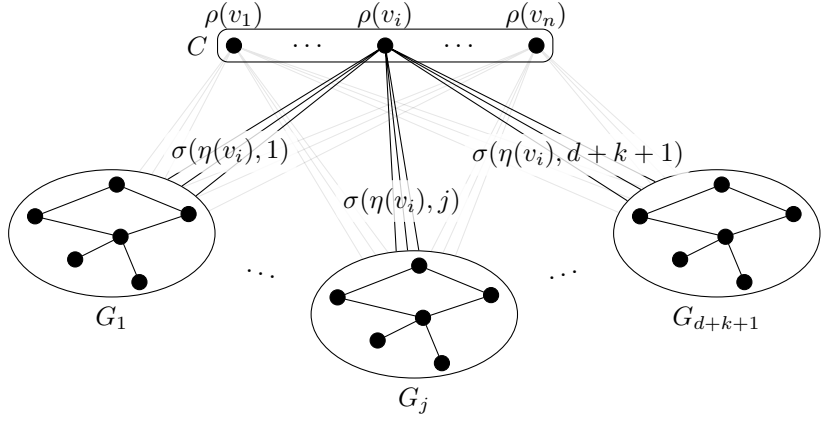

  \section{Power Dominating Set with Limited Propagation}
  \label{sec:w2l}

  Limiting the forcing rounds in \textsc{Power Dominating Set} yields a hierarchy of problems between
  \textsc{Dominating Set} and unrestricted \textsc{Power Dominating Set}.
  Since \textsc{Dominating Set} is $W[2]$-complete and \textsc{Power Dominating Set} is $W[P]$-complete,
  the two ends of this hierarchy have different parameterized complexity unless $W[2] = W[P]$.
  We show that the round-bounded variants not only interpolate between the problem statements but also between their expressive power.
  For fixed $\ell$, \textsc{$\ell$-round $k$-Power Dominating Set} captures exactly the $2\ell$-th level of the $W$-hierarchy.

  We prove the lower bound using the generalized forcing framework introduced above.
  This allows a somewhat natural representation of monotone normalized circuits as a forcing problem.
  The construction is inspired by the proof of $W[P]$-hardness of \textsc{Power Dominating Set} by Bläsius and Göttlicher~\cite{blasius_efficient_2025}.
  However, our gadgets are modified to exploit the round limit resulting in a hardness bound of $W[2\ell]$.

  To show that this bound is tight, we also provide a reduction from \textsc{$\ell$-round Generalized $k$-Forcing} to
  \textsc{Weighted Monotone $2\ell$-Normalized Circuit Satisfiability}.
  Since \textsc{$\ell$-round $k$-Power Dominating Set} is a special case, the same bound applies.
  This proves that \textsc{$\ell$-round $k$-Power Dominating Set} is $W[2\ell]$-complete.

  \subparagraph*{Lower Bound}
  To prove the lower bound, we reduce \textsc{WMNS[$t$]} to \textsc{$\ell$-round Generalized $k$-Forcing}.
  Our reduction uses a restricted variant of \textsc{WMNS[$t$]} which requires \AND-gates to have only one out-neighbor.
  This problem is equivalent to regular \textsc{WMNS[$t$]}, as seen in the following lemma.

  \begin{restatable}{lemma}{REwtwoelwmnsrestriction}
    \label{res:w2l:wmns-restriction}
    Every instance of \textsc{WMNS[$t$]} can be transformed in polynomial time to an equivalent instance where each
    \AND-gate has at most a single out-neighbor.
    The transformation preserves the weight of the assignment.
  \end{restatable}
  \begin{proof}
    Given a $t$-normalized monotone circuit $C$, we construct $C'$ as follows.
    Let $a$ be an \AND-gate in $C$ with two or more out-neighbors $o_1, \dots, o_\ell$.
    We replace $a$ by $\ell$ new \AND-gates $a_1, \dots, a_\ell$, such that $a_i$ has the same inputs as $a$ but only one output $o_i$.

    Circuit depth and normalization remain unaffected since the new gates are placed in the same layer as $a$.
    Since the circuit is normalized, no two \AND-gates are adjacent, and the size of the circuit only increases by a polynomial amount.

    Clearly, $a$ outputs true if and only if all inputs of $a$ are true.
    Since all $a_i$ have the same inputs, $a_i$ outputs true if and only if $a$ outputs true.

    Since the circuit and its input variables are otherwise unchanged, $C$ and $C'$ are equivalent.
  \end{proof}

  Observe that in a forcing problem, a marked vertex behaves similarly to an \AND-gate.
  In particular, with the $1$-forcing rule, a vertex can force precisely one vertex once all other vertices become marked.
  A natural transformation could thus try to simulate \True and \False gates by marked or unmarked vertices.
  Indeed, we can transform a circuit of fixed depth into an equivalent generalized forcing instance with bounded rounds
  with only minimal structural changes.
  Next, we show that this approach leads to a parameterized reduction from depth $2\ell$ circuit satisfiability to
  generalized forcing with $\ell$ rounds.
  This proves $W[2\ell]$-hardness of \textsc{$\ell$-round Generalized $k$-Forcing} and consequently of \textsc{$\ell$-round $k$-Power Dominating Set}.

  \begin{figure}
    \centering
    \begin{tikzpicture}
      \graph[] {
          {[nodes={node,shift={(-.5,0)}}, empty nodes] subgraph I_n[V={i1, i2, i3, i4}]}
        -> {[nodes=or, empty nodes] o1 1, o1 2, o1 3}
        -> {[nodes=and, empty nodes] a1 1/, a1 2/, a1 3/}
        -> {[nodes={or, shift={(0.5,0)}}] o2 1/, o2 2/}
        -> {/[not source,not target], /[and]} -> {/[not source], /};
          {i2, i3} -> {o1 1, o1 2};
          {o1 1, o1 2} -> {a1 2, a1 3, a1 1};
      };
      \graph[nodes={shift={(6, 0)}}] {
          {[nodes={node,selectable,shift={(-.5,0)}}, empty nodes] subgraph I_n[V={i1, i2, i3, i4}]}
        -- {[nodes=node, empty nodes] o1 1, o1 2, o1 3}
        -- {[nodes={node,marked, private leaf angle=0}, empty nodes] a1 1/, a1 2/, a1 3/}
        -- {[nodes={node,objective, target,shift={(0.5,0)}}] o2 1/, o2 2/};
          {i2, i3} -- {o1 1, o1 2};
          {o1 1, o1 2} -- {a1 2, a1 3, a1 1};
      };
    \end{tikzpicture}
    \caption{
      Example transformation from a normalized circuit to an equivalent $\ell$-round generalized 1-forcing instance.
      Nodes that initialize their neighborhood are represented as \legendnode{node,selectable},
      pre-marked nodes as \legendnode{node,marked} and target nodes as \legendnode{node,objective}.
      \label{fig:circuit-to-forcing}
    }
  \end{figure}
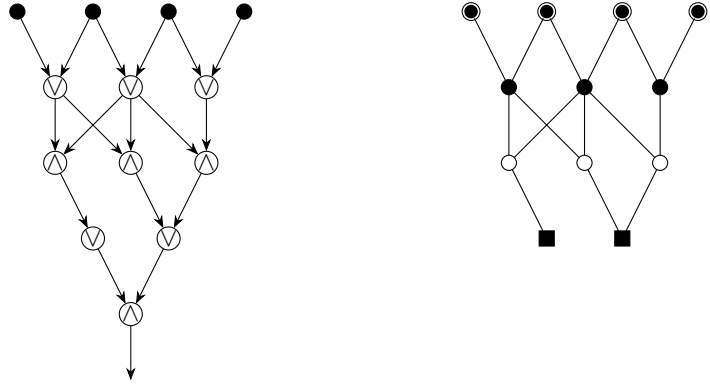

  \begin{restatable}{lemma}{REwtwoelhardness}
    \label{res:w2l-hardness}
    Every instance of \textsc{WMNS[$2\ell$]} can be transformed in polynomial time to an equivalent instance of \textsc{$\ell$-round Generalized $k$-Forcing}.
    The weight of a satisfying assignment in the circuit matches the solution size in the forcing instance.
  \end{restatable}
  \REwtwoelhardness*
  \begin{proof}
    Let $C$ be an instance of \textsc{WMNS[$2\ell$]}.
    By \cref{res:w2l:wmns-restriction}, we may assume that every \AND-gate in $C$ has at most one out-neighbor.

    We construct an equivalent extended generalized forcing instance $I$ from $C$ as illustrated in \cref{fig:circuit-to-forcing}.
    By \cref{res:generalized-forcing-extension}, this is equivalent to constructing a basic instance.
    First, we remove the output \AND-gate and replace all remaining gates and inputs with graph vertices.
    Denote the three resulting vertex sets by $V_{\text{in}}$, $V_\OR$ and $V_\AND$.
    We keep all edges of the circuit as undirected edges.
    For all input vertices $v \in V_{\text{in}}$, set the initialization $\init(v) = N[v]$.
    All other vertices have an empty initialization.
    We pre-mark all \AND-vertices.
    To account for the forcing threshold $k$, we attach $k-1$ new leaves to each \AND-vertex.
    All vertices in $I$ except the last layer of \OR-vertices are optional.
    As a convention, we refer to the input vertices as layer 1, to the first layer of \OR-gates as layer 2 and so on.

    The key idea of the construction is that the \True \OR-gates in layer $2j$ correspond to the \OR-vertices that become marked in round $j$.

    Now, we show that $C$ has a solution of size $d$ if and only if $I$ has a solution of size $d$.

    \pfdir{Satisfying assignment to forcing set:}
    Let $X$ be a satisfying assignment of $C$ with $\abs{X} = d$ and let $S$ be the set of input vertices corresponding
    to \True variables.
    We show by induction that $S$ is a solution of $I$.

    The true inputs in $X$ set their adjacent \OR-gates in the first layer to \True.
    This corresponds to the \OR-vertices in the first layer of the graph which are marked in initialization.
    Next, the \AND-gates in the following layer are evaluated.
    They output \True if all of their inputs are \True.
    This is replicated by the pre-marked \AND-vertices $v \in P$.
    These have $k$ neighbors left, the out neighbor and $k-1$ leaves.
    Once all in-neighbors are marked, $v$ can force its leaves and the next gate.
    Inductively, the last layer of \OR-gates becomes marked after precisely $\ell$ rounds of propagation.

    \pfdir{Forcing set to satisfying assignment:}
    Let $S$ be a solution of $I$.
    Since $\init(v) = \emptyset$ for all $v \notin V_{\text{in}}$ , selecting a non-input vertex has no effect on the initially marked vertices.
    Hence, $S \cap V_{\text{in}}$ is also a solution.
    Let $X$ be the assignment where an input variable is true if and only if its corresponding vertex is in $S$.
    Clearly, $X$ has weight at most $d$.
    We show that $X$ is a satisfying assignment of $C$.

    First observe that no \OR-vertex in layer $2i$ can become marked in fewer than $i$ rounds.
    This is enforced by the layered structure since forcing cannot skip a layer.

    We now use induction over the propagation rounds to show that the \OR-vertices that become marked in round $i$
    correspond to the \True \OR-gates in layer $2i$ of the circuit.
    More precisely, we show that if an \OR-vertex in layer $2i$ becomes marked in round $i$,
    its corresponding \OR-gate evaluates to \True.
    We count layers starting from the inputs, which are in layer one.

    In the first layer, exactly those \OR-vertices that are adjacent to a selected input are marked.
    This corresponds to the circuit where the \OR-gates that are adjacent to a \True input evaluate to \True.

    Now assume that in round $i$, every marked \OR-vertex in layer $2i$ corresponds to a \OR-gate that evaluates to \True.
    The next layer consists of \AND-vertices and \AND-gates, respectively.
    Each pre-marked \AND-vertex $v$ in layer $2i+1$ is adjacent to $k-1$ leaves and an \OR-vertex $w$ in layer $2i+2$.
    None of these vertices can be marked in round $i$.
    To force $w$, $v$ must have at most $k$ unmarked neighbors.
    Thus, $v$ can only force $w$ in round $i+1$ if all its in-neighbors, which are \OR-gates, are marked in round $i$.
    By our assumption, all corresponding \OR-gates evaluate to \True.
    Hence, the \AND-gate corresponding to $v$ evaluates to \True if and only if $v$ forces $w$.
    Consequently, an \OR-vertex in layer $2i+2$ can only become marked in round $i+1$ if its corresponding \OR-gate outputs \True.
    This proves the induction claim.

    Since $S$ is a solution, all \OR-vertices in layer $2\ell$ are marked by round $\ell$ and the layered structures ensures
    that they cannot become marked earlier.
    Thus, by the induction, the corresponding \OR-gates all evaluate to \True.
    Since these are exactly the in-neighbors of the output \AND-gate, $C$ evaluates to \True.
    Hence, $X$ is a satisfying assignment.
  \end{proof}

  \subparagraph*{Upper Bound}
  We now give an upper bound on the parameterized complexity by reducing \textsc{Generalized Forcing($k,\ell$)} to
  \textsc{Weighted Monotone Circuit Satisfiability($2\ell$)} for fixed integers $k$ and $\ell$.

  The core of this reduction is the propagation gadget illustrated in \cref{fig:propagation-gadget}.
  The idea is that we use a layer of \OR-gates to model which vertices are marked after $j$ propagation rounds.
  This layer is followed by a layer of \AND-gates which in turn model the propagation condition for each vertex in
  each neighborhood.
  Recall that a vertex $u$ can propagate to its neighbors if it is marked and has at most $k$ unmarked neighbors.
  For each $K \subseteq N(u)$ with $\abs{K} \leq k$, we add an \AND-gate that has all vertices in $N[V] \setminus K$ as inputs and verifies whether $u$ and all neighbors outside $K$ are marked.
  The \OR-gates of the next round collect the outputs of the \AND-gates of each vertex in each neighborhood gadget.

  To model the initialization, we add an input node for every vertex in the graph.
  The input of a node $u$ is then connected to the \OR-gates of $I(u)$ in the first propagation layer.
  Then, we add another \AND-gate and connect it to the last layer of \OR-gates.
  This \AND-gate serves as the output of the circuit.

  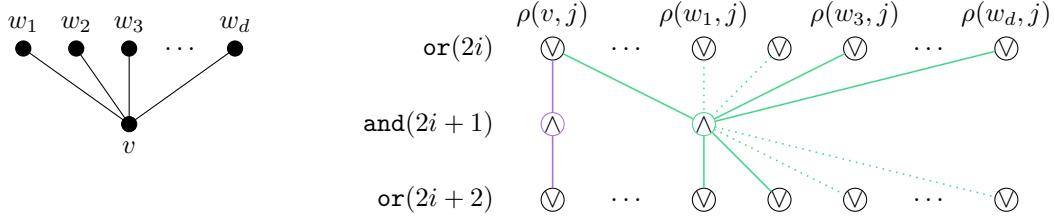
\begin{figure}
    \centering
    \begin{tikzpicture}[notedge/.style={dotted}]
      \begin{scope}[shift={(-7, 0)}]
        \graph[empty nodes,branch right=0.7] {
            {[nodes=node] w1["$w_1$"],w2["$w_2$"],w3["$w_3$"],/[dots],wd["$w_d$"]} -- {/[hidden], /[hidden], v["$v$" below,node]};
        };
      \end{scope}
      \graph[empty nodes,edges=semithick]{
          {rvj[or,"{$\rho(v,j)$}"], /[dots], w1j[or,<notedge,"{$\rho(w_1,j)$}"],w2j[or,<notedge],w3j[or,"{$\rho(w_3,j)$}"],/[hidden,dots],wdj[or,"{$\rho(w_d,j)$}"]}
        -- {av[and,draw=violet6,<violet6,>violet6], /[hidden], apj[and,draw=teal8,<teal8,>teal8]}
        -- {rvjj[or], /[dots], w1jj[or],w2jj[or],w3jj[or,>notedge],/[hidden,dots],wdjd[or,>notedge]};
        rvj -- apj;
      };
      \node[left=0.5 of rvj] {$\OR(2i)$};
      \node[left=0.5 of av] {$\AND(2i+1)$};
      \node[left=0.5 of rvjj] {$\OR(2i+2)$};
    \end{tikzpicture}
    \caption{
      Graph on the left and part of the propagation gadget on the right.
      The gadget is shown for forcing threshold $k=2$ and the vertex set $\set{w_1, w_2}$.
      Dotted lines represent non-edges.
      The monotonicity gadget for $v$ connects $\rho(v,j)$ to its counterpart $\rho(v, j+1)$.
    }
    \label{fig:propagation-gadget}
  \end{figure}

  \begin{restatable}{lemma}{REwtwoelcontained}
    \label{res:w2l-contained}
    Every instance of \textsc{$\ell$-round Generalized $k$-Forcing} with finite round limit $\ell$ can be transformed in
    polynomial time into an equivalent monotone $2\ell$-normalized circuit.
    The minimum solution size of the forcing instance is preserved as the weight of a satisfying assignment of the circuit.
  \end{restatable}
  \begin{proof}
    Let $I$ be an instance of \textsc{$\ell$-round Generalized $k$-Forcing} with graph $G$.
    We construct a monotone circuit $C$ with one input variable $x_v$ for each vertex $v \in V(G)$.
    Selecting $v$ in $I$ corresponds to setting $x_v = \True$.
    The inputs form the first layer of the circuit.

    The idea here is to simulate the propagation rounds through consecutive \OR- and \AND-layers in the circuit.
    If an \OR-gate in layer $2i$ outputs \True, this means that the corresponding vertex is marked in round $i$.
    We then use \AND-gates in the following \AND-layer to simulate the condition of the forcing rule for each sub-neighborhood of each vertex with at most $k$ vertices.
    The neighborhood of the input gates corresponds to the initialization and the output \AND-gate connects to the last layer of \OR-gates.
    Since the \True \OR-gates in each layer $2i$ correspond to the set of marked vertices in round $i$,
    every forcing instance with at most $\ell$ rounds can be simulated by a circuit with $2\ell$ layers.

    For every round $i \in \set{1, \dots, \ell}$ and every vertex $v$, we introduce an \OR-gate $\rho(v, i)$.
    This gate represents whether $v$ is marked in round $i$.
    The first layer of \OR-gates represents the initially marked set, thus add edges from each $x_v$ to $\rho(\init(v), 1)$.

    We simulate the forcing process using the \AND-gates.
    For each propagation round $i \in \set{2, \dots, \ell}$, and each vertex $v$ we add two kinds of gates.
    First, we add a monotonicity gate that connects $\rho(v, i-1)$ to $\rho(v, i)$.
    Then, for each sub-neighborhood $K \subseteq N(v)$ with $\abs{K} \leq k$, we add a forcing gate $\sigma(v, K, i)$ with inputs $\rho(N[v] \setminus K, i-1)$ and outputs $\rho(K, i)$.
    Finally, add one output \AND-gate adjacent to last layer of \OR-gates.

    The circuit has one \OR-layer and one \AND-layer for each round and gates have neighbors only in adjacent layer.
    Its depth is thus $2\ell$.
    Since $k$ is fixed, the number of forcing gates is polynomial in the vertex degree and thus the circuit size is polynomial in the graph size.

    \pfdir{Correctness of the construction}
    Let $S$ be a vertex selection and let $X$ be the corresponding assignment.
    We show by induction that the vertices marked by round $i$ are exactly those vertices $v$ for which $\rho(v, i)$ is \True.

    In the first round, our construction ensures that marked vertices $v$ correspond to some $\rho(v, 1)$
    that is adjacent to a \True input.

    Now assume that the correspondence holds up to some round $i$.
    We show that it continues to hold for round $i+1$.
    For each marked vertex, the monotonicity \AND-gate ensures that $\rho(v, i+1)$ is \True.
    If $v$ can force in round $i$, it has at most $k$ unmarked neighbors $K \subseteq N(v)$.
    Since $\rho(N[v] \setminus K)$ is \True by our assumption, $\sigma(v, K, i+1)$ outputs \True to all $\rho(K, i+1)$.
    Every marked vertex in round $i + 1$ thus corresponds to a \True \OR-gate.

    Conversely, if an \OR-gate $\rho(v, i)$ in layer $2i$ evaluates to \True, then its counterpart in layer $2i+2$ outputs \True through the monotonicity gate.
    If some $\rho(v, i+1)$ is \True but $\rho(v, i)$ is \False, there must be a \True forcing gate $\sigma(w, K, i+1)$ with $v \in K$.
    Then, by our assumption, the vertex $w$ is marked and has at most $k$ unmarked neighbors.
    Hence, all neighbors $K$, in particular $v$, are forced.
    The correspondence thus holds in both directions round $i+1$.

    \pfdir{Solution Equivalence}
    Now assume that $S$ is a solution.
    Then, all vertices are marked in round $\ell$ and, by our induction, the last layer of \OR-gates is \True.
    Thus the output \AND-gate evaluates to \True and the circuit is satisfied.

    Conversely, let $X$ be a satisfying assignment.
    Since the output gate is \True, the last layer of \OR-gates must also be \True.
    Then all vertices are marked in round $\ell$.
  \end{proof}

  The upper and lower bound match, thus we get the following result.
  \begin{restatable}{theorem}{REwtwoelcomplete}
    \label{res:w2l-complete}
    For fixed integers $k \geq 1$ and $\ell \geq 1$, \textsc{$\ell$-round Generalized $k$-Forcing} and
    \textsc{$\ell$-round $k$-Power Dominating Set} are $W[2\ell]$-complete.
  \end{restatable}

  \section{Hardness of unbounded \texorpdfstring{$k$-Forcing}{k-Forcing}}
  \label{sec:k-forcing}
  While \textsc{Zero Forcing} is fixed-parameter tractable, this has not been shown for \textsc{$k$-Forcing}.
  In this section we show that $k$-Forcing is in fact very unlikely to be fixed-parameter tractable.
  We show that raising the forcing threshold to $k=2$ already yields $W[P]$-hardness for \textsc{$k$-Forcing}.

  As an intermediate step, we extend instances of \textsc{$k$-Forcing} with pre-marked vertices.
  This yields another interesting result: when allowing pre-marked vertices, \textsc{Zero Forcing} becomes $W[P]$-hard, strengthening the previously known $W[2]$-hardness result~\cite{cazals_power_2019}.
  We also derive another bounded-round result: \textsc{Pre-Marked $\ell$-round $k$-Forcing} is hard for $W[2(\ell-2)]$.

  \begin{restatable}{lemma}{REpdstokforcing}
    \label{res:gf-to-premarked}
    Every instance of \textsc{$\ell$-round Generalized $k$-Forcing} where $\ell$ may be unbounded can be transformed in
    polynomial time into an equivalent extended instance of \textsc{$k$-Forcing} with pre-marked vertices.
    If $\ell$ is finite, the round limit increases to $\ell+2$; unbounded rounds remain unbounded.
    The transformation preserves the solution size.
  \end{restatable}
  \REpdstokforcing*
  \begin{proof}
    Let $I$ be an instance of \textsc{$\ell$-round Generalized $k$-Forcing} with graph $G$, initialization $\init$ and solution size $d$.
    We construct an equivalent extended instance $I'$ of \textsc{$\ell+2$-round $k$-Forcing} or unbounded \textsc{$k$-Forcing} with implication arcs and optional vertices.
    By \cref{res:extended-zf-instances}, we may use optional vertices and implication arcs since these later can be eliminated.
    The only remaining instance extension is the set of pre-marked vertices.

    The construction is similar to the proof of \cref{res:pds-hardness}.
    However, with $k$-forcing, we need to modify the construction since we cannot use domination for marking the initial set.
    Instead, we rely on implication arcs to simulate the initialization.
    The disadvantage of this approach is that implication arcs introduce additional propagation rounds.
    With bounded propagation rounds, this results in less tight hardness bounds.
    In more detail, we construct $G'$ as follows.

    We create $d+1$ disjoint copies $G_j$ of $G$.
    For each vertex $v \in V(G)$, denote by $\sigma(v, j)$ its counterpart in $G_j$.
    Since there are $d+1$ copies, there will always be one that has no selected vertex.

    Then insert $\abs{V}$ new optional proxy vertices $C = \set{\rho(v) \mid v \in V(G)}$.
    We for each $v \in V(G)$ and each $w \in \init(v)$ add an implication arc from $\rho(v)$ to each $\sigma(w, j)$.
    These simulate the initialization in the copies.

    \pfdir{Generalized Forcing to $k$-Forcing}
    Let $S$ be a solution of $I$.
    We show that $S' = \rho(S)$ is a solution of $I'$.

    In the initialization, $S'$ marks exactly the vertices in $S'$, i.e. $\rho(S)$.
    In each copy $G_j$, the implication arcs thus mark the vertices $\sigma(w, j)$ with $w \in \init(S)$.
    The copies are isomorphic to $G$, thus, starting with a corresponding initial marked set,
    forcing in each copy $G_j$ corresponds exactly to forcing in $G$.
    Since $S$ is a solution all vertices in $G$ become marked in $\ell$ rounds, hence,
    the vertices in each $G_j$ become marked in $\ell+2$ rounds.

    The only remaining vertices in $G'$ are the proxies in $C$, which are optional.
    Thus $S'$ is a solution of $I'$.

    \pfdir{$k$-Forcing to Generalized Forcing}
    Let $S'$ be a solution of $I'$.
    As a first step, we show that $S^\ast = S' \cap C$ is also a solution of $I'$.
    We then show that $S = \rho^{-1}(S^\ast)$ is a solution of $I$

    Since there are $d+1$ copies, at least one copy $G_{j^\ast}$ contains no selected vertex.
    The only way vertices in this copy can become marked is through the implication arcs from $C$.
    Accordingly, no vertex in $G_{j^\ast}$ becomes marked before round $3$.
    Denote by $M$ the set of vertices marked in $G_{j^\ast}$ through implication arcs.
    Since $S'$ is a solution, and no vertex in $M$ is marked before round 3, there are only $\ell$ rounds left to mark all vertices in $G_{j^\ast}$.
    Thus, $M$ must be an $\ell$-round $k$-forcing set of $G_{j^\ast}$.
    By isomorphism, $\sigma^{-1}(M, j^\ast)$ is also an $\ell$-round $k$-forcing set of $G$.

    Since each $v \in C$ has only implication arcs, $v$ can only become marked if it is selected.
    Then the vertices in $G_{j^\ast}$ are marked starting only with $S^\ast = S' \cap C$.
    The vertices in $M$ thus correspond exactly to $\init(S)$.
    Hence, $S$ is a solution of $I$.
  \end{proof}

  Combining \cref{res:w2l-complete,res:gf-to-premarked} allows us to derive hardness bounds for \textsc{Pre-Marked $\ell$-round $k$-Forcing} with bounded and unbounded propagation.
  Since this problem is a special case of \textsc{$\ell$-round Generalized $k$-Forcing}, the upper bounds still apply.
  \begin{restatable}{theorem}{REpremarkedzfhardness}
    \label{res:pre-marked-zf-hardness}
    For fixed $\ell \geq 3$ and fixed $k \in \Natural$, the problem \textsc{Pre-Marked $\ell$-round $k$-Forcing} is $W[2(\ell-2)]$-hard and in $W[2\ell]$.
    \textsc{Pre-Marked $k$-Forcing} with unbounded propagation is $W[P]$-complete.
  \end{restatable}
  The bounds for \textsc{Pre-Marked $\ell$-round $k$-Forcing} leave a gap between $W[2(\ell-2)]$ and $W[2\ell]$.
  We do not expect the upper bound to be tight.
  Indeed, for $\ell=1$, there are no forcing rounds, and thus \textsc{Pre-Marked $1$-round $k$-Forcing} asks the trivial question whether $S \cup P = V(G)$.
  In contrast, \textsc{$1$-round Power Dominating Set} is identical to \textsc{Dominating Set} and is therefore $W[2]$-complete.
  This suggests that bounded pre-marked forcing captures a lower level of the $W$-hierarchy than bounded power domination.

  With $k \geq 2$, and an unbounded round budget, we can also eliminate the pre-marked vertices by connecting them through a $k$-ary tree.
  Selecting the tree root will eventually mark the pre-marked vertices.
  However, since the size of the tree depends on the input, this does not apply to any fixed round budget $\ell$.
  In the following lemma, we show the construction for unlimited rounds.

  \begin{restatable}{lemma}{REkforcingpremarked}
    \label{res:k-forcing-premarked}
    Every instance of \textsc{Pre-Marked $k$-Forcing} with fixed $k > 1$ can be transformed into an equivalent instance of \textsc{$k$-Forcing}.
    The transformation increases the solution size by one.
  \end{restatable}
  \begin{proof}
    Let $I$ be an instance of \textsc{Pre-Marked $k$-Forcing} with graph $G$ and pre-marked vertices $P$.
    We construct an equivalent instance \textsc{$k$-Forcing} instance $I'$.

    The idea is that, with unbounded propagation, we do not need pre-marked vertices to actually become marked in the first round,
    since forcing is monotone.
    Furthermore, with $k > 1$, selecting the root of a $k$-ary tree of height $h$ will mark all leaves of the tree in at most $h$ rounds.

    In more detail, we construct $G'$ as follows.
    We construct a $k$-ary tree $T$ whose leaves are $P$ and $k+1$ new vertices $P'$.
    Denote the root of $T$ by $r$.
    The tree $T$ does not have to be complete; if $\abs{P'}$ is not a power of $k$, we create a tree with fewer leaves.
    The tree has height $h \in O(\log_k \abs{P'})$.
    Last, connect the vertices in $P \cup P'$ into a clique.

    \pfdir{Pre-marked to no pre-marked: }
    Let $S \subseteq V(G)$ such that $S \cup P$ is a $k$-forcing set of $G$.
    We show that $S' = S \cup \set{r}$ is a $k$-forcing set of $G'$.

    Since $r$ has only $k$ neighbors, $r$ can mark these neighbors by forcing.
    These neighbors in turn now have only $k$ unmarked neighbors left and can propagate.
    Inductively, all vertices in $T$ become marked in $\log_k{\abs{P}}$ rounds, in particular all vertices in $P$.

    Thus, after $\log{\abs{P}}$ rounds, all vertices in $V(G)$ have the same unmarked neighbors as in $G$ and
    propagation becomes equivalent.

    \pfdir{No pre-marked to pre-marked: }
    Let $S'$ be a $k$-forcing set of $G'$.
    We show that $S = S' \cap V(G)$ is a solution of the original instance $I$.

    First, assume that $S'$ contains no vertex from $T$.
    Then, only vertices in $P$ can be marked by forcing from $V(G)$.
    Our construction, ensures that there are at least $k+1$ vertices in $P'$ that are not adjacent to any vertex
    in $V(G)$ and the vertices in $P \cup P'$ form a clique.
    Thus, each vertex in $P$ has at least $k+1$ unmarked neighbors and cannot propagate,
    leaving all other vertices in $T$ unmarked.
    Therefore, at least one vertex in $T$ must be selected.
    We can replace all selected vertices in $T$ by the root $r$ without increasing the solution size, since,
    as we saw above, selecting $r$ causes all of $T$ to be marked by propagation.

    If $r$ is selected, all vertices in $P$ become marked after $\log_k \abs{P}$ rounds.
    Since we have unlimited propagation rounds, $k$-forcing in $G'[V(G)]$ then becomes equivalent to
    $k$-forcing in $G$ with pre-marked vertices $P$.
  \end{proof}

  Since the solution size is only modified by a constant, we have a parameterized reduction and thus the hardness of \textsc{Generalized $k$-Forcing} transfers to \textsc{$k$-Forcing} with $k > 1$.

  \begin{restatable}{theorem}{REkforcingwp}
    \label{res:k-forcing-wp}
    When parameterized by the solution size, \textsc{$k$-Forcing} is $W[P]$-complete for fixed $k>1$.
  \end{restatable}

  \section{\texorpdfstring{$\ell$-round $k$-Forcing}{l-round k-Forcing} is Fixed Parameter Tractable}
  \label{sec:k-forcing-fpt}
  Bounding the round limit in \textsc{$\ell$-round Power Dominating Set} interpolates in complexity between \textsc{Dominating Set} and \textsc{Power Dominating Set}.
  One might expect to see a similar effect with \textsc{$\ell$-round $k$-Forcing}.
  However, we show that bounding the propagation rounds immediately makes the problem fixed-parameter tractable.
  For \textsc{$\ell$-round Zero Forcing}, an FPT-algorithm was already given by Aazami~\cite{aazami_hardness_2008}.
  We present a new FPT-algorithm for all finite round limits $\ell$ and forcing thresholds $k$.
  Unlike the previous algorithm, our algorithm does not use a tree-decomposition.
  Instead, we show that for fixed solution size $d$ and limited rounds, the vertex count of all yes-instances is bounded by a constant.
  This yields a trivial polynomial kernel for \textsc{$\ell$-round $k$-Forcing}.
  Whether \textsc{Zero Forcing} with unbounded propagation admits a polynomial kernel remains unclear.

  \begin{restatable}{theorem}{RElroundkforcing}
    \label{res:l-round-k-forcing}
    For fixed integers $\ell$ and $k$, \textsc{$\ell$-round $k$-Forcing} is fixed parameter tractable
    and has a polynomial kernel when parameterized by the solution size.
  \end{restatable}
  \begin{proof}
    The key observation is that in each propagation round, a single marked vertex can force at most $k$ new vertices.
    This means that a given initially marked set $\init(S)$ can mark at most $\abs{\init(S)} \cdot (k+1)^{\ell-1}$ vertices.
    With limited propagation rounds and ZF-style initialization $\init(S) = S$, when fixing the solution size $\abs{S} = d$ this is a constant limit on the number of vertices that can be forced.
    Any instance exceeding this limit is thus trivially a no-instance.
    By exhaustive search among all possible solutions of the instances with at most $d (k+1)^{\ell-1}$ vertices,
    we obtain an FPT-algorithm.

    We also obtain a trivial kernel by mapping all larger instances to a fixed no-instance.
  \end{proof}

  \subparagraph*{A brief detour: Kernels for Connected Forcing}
  This positive kernel result contrasts with a kernel lower bound for another variant, \textsc{Connected Zero Forcing}~\cite{brimkov_complexity_2017,davila_bounds_2018}.
  In this variant, a solution must also induce a connected graph.
  We show that \textsc{Connected Zero Forcing} does not admit a polynomial kernel unless $\text{coNP} \subseteq \text{NP/poly}$ by reduction from \textsc{Red-Blue Dominating Set (RBDS)}.
  Given a bipartite graph with partitions $R$ and $B$, RBDS asks whether there is a set $D \subseteq R$ with $\abs{D} \leq d$ that dominates $B$.
  RBDS is NP-complete and, when parameterized by $d + \abs{B}$, it does not admit a polynomial kernel unless $\text{coNP} \subseteq \text{NP/poly}$~\cite{dom_kernelization_2014}.
  \begin{restatable}{proposition}{REconnectedforcingkernelbound}
    \label{res:connected-forcing-kernel-bound}
    When parameterized by the solution size, \textsc{Connected Zero Forcing} does not admit a polynomial kernel unless $\text{coNP} \subseteq \text{NP/poly}$.
  \end{restatable}
  \REconnectedforcingkernelbound*
  \begin{proof}
    Let $I$ be an instance of RBDS with graph $G$ and $V(G) = R \cup B$.
    We construct an instance $I'$ of \textsc{Connected Zero Forcing} with graph $G'$ as follows.
    First, we add a new vertex $\tilde{u}$ adjacent to all vertices in $R$.
    Then, we attach $2$ leaves to $\tilde{u}$ and to each vertex in $B$.
    The leaves enforce that these vertices are in a connected forcing set, since, if $v$ is a cutvertex that separates $3$ or more components, $v$ must be in every connected forcing set.
    Let $r_1, \dots, r_p$ be the red vertices.
    We add a path $Q = \Path{\ell_1, r_1, s_1, r_2, s_2, \dots, r_p, \ell_p}$ with new vertices $\ell_1, \ell_p, s_1, \dots, s_{p-1}$.
    Note that $Q$ is an induced path.

    We show that $G'$ has a connected forcing set of size $d' = d + 2 \abs{B} + 3$ if and only if $I$ has a red-blue dominating set $D$ of size at most $d$.
    The idea is that a connected forcing set must select $\tilde{u}$ and all vertices in $B$ along with one leaf of each.
    Therefore, to obtain a connected solution one neighbor of each $b \in B$ must be selected.

    Let $D$ be a red-blue dominating set of $I$.
    Then selecting $D$, all vertices in $B$, $\tilde{u}$, one leaf of each vertex in $B$ and of $\tilde{u}$, and a neighbor on $Q$ of some vertex in $D$ yields a connected forcing set $S$ of $I'$.
    Clearly, $S$ has size $\abs{D} + 2 \abs{B} + 3$; it is easy to see that $G[S]$ is connected.
    Since we select two adjacent vertices in $Q$ and no path vertex has an unmarked neighbor outside $Q$, all vertices in $Q$ become marked.
    Then $\tilde{u}$ and the vertices in $B$ have only one unmarked neighbor left which they can force.

    Conversely, let $S$ be a connected forcing set of $I'$ with $\abs{S} \leq d'$.
    Since all vertices in $B$ and $\tilde{u}$ each have two leaves, one leaf must be selected; otherwise there remain two unmarked neighbors and the leaves can never become marked.
    Furthermore, since the vertices in $B$ and $\tilde{u}$ are cutvertices with $3$ components, they must be selected.
    Moreover, at least one vertex in $Q \setminus R$ must be selected;
    otherwise, all vertices in $R$ have two unmarked neighbors in $Q \setminus R$ and cannot force.
    Therefore, there are $2 \abs{B} + 3$ selected vertices outside $R$ and $\abs{S \cap R} \leq d$.
    Assume that $S \cap R$ is no red-blue dominating set, i.e. there is $b \in B$ that has no neighbor in $S \cap R$.
    Then $G[S]$ is disconnected, resulting in a contradiction.

  \end{proof}

  \section{Conclusion}
  \label{sec:conclusion}

  We classify variants of \textsc{Zero Forcing} and \textsc{Power Dominating Set} along the choices of initialization, round limit and forcing threshold.
  Our classification shows that the parameterized complexity of the resulting problem is determined by the interaction of all three choices.
  Notably, introducing a fixed limit on the propagation rounds of \textsc{Power Dominating Set} captures the even levels of the $W$-hierarchy
  and increasing the forcing threshold in \textsc{Zero Forcing} to $k=2$ causes a jump in complexity from FPT to $W[P]$-completeness.
  Even with higher propagation threshold, \textsc{$\ell$-round $k$-Forcing} remains fixed parameter tractable and admits a polynomial kernel.
  In fact, we show that for fixed solution size and bounded propagation, the solution space collapses to instances with at most a constant number of vertices.
  Whether standard \textsc{Zero Forcing} admits a polynomial kernel remains open; in particular, we show that another variant, \textsc{Connected Zero Forcing}, does not admit a polynomial kernel under standard assumptions.

  Our hardness results require fixed forcing thresholds and round limits.
  If the forcing threshold is part of the input, we expect most results to hold, despite required changes to the constructions.
  However, the FPT algorithm and kernelization for \textsc{$\ell$-round $k$-Forcing} rely on $k$ being fixed and cannot easily be transferred.
  It is easy to see that the unbounded case reduces to the variable-$\ell$ case by setting $\ell = \abs{V(G)}$.

  In graphs with bounded treewidth several of the variants we studied were shown to be solvable in polynomial time~\cite{guo_improved_2005,aazami_hardness_2008}.
  However, \textsc{Generalized Forcing} with its arbitrary initialization is too expressive to admit a simple adaptation of these algorithms.
  Indeed, we can model every instance of \textsc{Dominating Set} as an instance of \textsc{Generalized Forcing} with a graph of isolated vertices which has treewidth 0.

  We introduced the generalized forcing framework as a reduction tool but the problem may be of interest in its own right.
  Under suitable restrictions on the initialization, other algorithmic results for \textsc{Zero Forcing} and \textsc{Power Dominating Set} may transfer to \textsc{Generalized Forcing}.
  In particular, solving the problem through the implicit hitting set approach should yield a straightforward solution technique~\cite{brimkov_computational_2019,smith_new_2022,blasius_efficient_2025}.

  Not all problem variants can be classified only by initialization, round limit and forcing threshold.
  \textsc{Throttling} minimizes not the solution size but the combination of solution size and round limit~\cite{brimkov_power_2019}.
  In \textsc{Skew Forcing}, unmarked vertices with sufficiently few unmarked neighbors may force, too~\cite{ima-isu_research_group_on_minimum_rank_minimum_2010}.
  Understanding how these additional dimensions fit into the complexity landscape remains an interesting question for future work.

  Overall, we show that forcing problems are sensitive to small changes in the problem statement.
  Indeed, by choosing different initializations, forcing threshold and round limits, one obtains forcing problems ranging from FPT to complete problems in every even level $W[2\ell]$ of the $W$-hierarchy up to $W[P]$.
  Our work provides a systematic view of the interaction between initialization, round limit and forcing threshold in the parameterized complexity of forcing problems.

  \phantomsection
  \addcontentsline{toc}{section}{References}
  \bibliography{bibliography}

  \end{document}

%% file: settings/graphics.tex
\usepackage{ninecolors}
\usepackage{tikz,pgfplots}
\pgfplotsset{compat=1.18}
\usetikzlibrary{
  arrows.meta,
  backgrounds,
  calc,
  decorations.pathreplacing,
  decorations.markings,
  fit,
  graphs,
  graphs.standard,
  positioning,
  quotes,
  shapes.geometric,
}
\usepgflibrary{plotmarks}

\tikzset{
  >/.tip=Stealth,
  ->-/.style={postaction=decorate,decoration={markings,mark={at position #1 with \arrow{>}}}},
  ->-/.default=0.5,
  -<-/.style={postaction=decorate,decoration={markings,mark={at position #1 with \arrow{<}}}},
  -<-/.default=0.5,
  node/.style={fill,draw, circle, inner sep=0, outer sep=0, minimum size=2mm},
  small node/.style={node, minimum size=1.3mm},
  dist/.style={node,fill=none,star, star points=11, star point height=1.5pt},
  marked/.style={fill=none},
  objective/.style={rectangle},
  selectable/.style={double},
  and/.style={node,fill=white,node contents={$\wedge$}},
  or/.style={node,fill=white,node contents={$\vee$}},
  t/.style={draw=none,fill=none},
  c/.style={coordinate},
  hv/.style={rounded corners,to path={-| (\tikztotarget)}},
  vh/.style={rounded corners,to path={|- (\tikztotarget)}},
  private leaf length/.initial=3mm,
  private leaf spread/.initial=30,
  private leaf angle/.initial=60,
  private leaf count/.initial=4,
  private leaves/.style={
    /utils/exec={
      \pgfmathtruncatemacro{\PLn}{\pgfkeysvalueof{/tikz/private leaf count} - 1}
      \xdef\PLn{\PLn}
      \pgfmathsetmacro{\PLa}{\pgfkeysvalueof{/tikz/private leaf angle}}
      \xdef\PLa{\PLa}
      \pgfmathsetmacro{\PLs}{\pgfkeysvalueof{/tikz/private leaf spread}}
      \xdef\PLs{\PLs}
    },
    append after command={
      \pgfextra{
      \draw[very thin]
      foreach \i in {1,...,\pgfkeysvalueof{/tikz/private leaf count}}{
        let \n1={\PLa + (\i - 1 - (\PLn) / 2) * \PLs / \PLn} in
        (\tikzlastnode) -- ++(\n1:\pgfkeysvalueof{/tikz/private leaf length})
      };
    }
  }},
}
\tikzset{
  dots/.style={node contents={$\dots$}, draw=none, fill=none},
  vdots/.style={node contents={$\cdots$}, rotate=90,circle,draw=none,fill=none},
  graphs/.cd,
  hidden/.style={not source, not target},
  dots/.style={hidden,/tikz/dots},
  vdots/.style={hidden,/tikz/vdots},
  grow down,
  branch right,
}

\newcommand{\legendnode}[1]{\tikz \node[#1] {};}

\makeatother

\newcommand{\DeclareNPointedStar}[3][0.8]{%
  \begingroup
  \count@=0\relax
  \dimen@=\dimexpr360pt/#3\relax

  \toks@={}%

  \loop
  \ifnum\count@<#3\relax
  \edef\pgf@temp{%
    \noexpand\pgfpathlineto{%
      \noexpand\pgfpointpolar%
      {\strip@pt\dimexpr90pt+\count@\dimen@\relax}%
      {\noexpand\pgfplotmarksize}}%
    \noexpand\pgfpathlineto{%
      \noexpand\pgfpointpolar
      {\strip@pt\dimexpr90pt+\count@\dimen@+\dimen@/2\relax}%
      {\noexpand#1*\pgfplotmarksize}%
    }%
  }%

  \toks@=\expandafter\expandafter\expandafter{%
    \expandafter\the\expandafter\toks@
    \pgf@temp
  }%

  \advance\count@ by1\relax
  \repeat

  \edef\pgf@temp{%
    \endgroup
    \noexpand\pgfdeclareplotmark{#2}{%
      \noexpand\pgfpathmoveto{\noexpand\pgfqpoint{0pt}{\noexpand\pgfplotmarksize}}
      \the\toks@
      \noexpand\pgfpathclose%
      \noexpand\pgfusepathqfillstroke%
    }%
  }%
  \pgf@temp
}

\makeatother